\documentclass[journal,onecolumn,12pt]{IEEEtran}

\usepackage{amsmath}
\usepackage{amsthm}
\usepackage{amssymb}
\usepackage{bm}
\usepackage{xspace}
\usepackage{xcolor}
\usepackage{graphicx}
\usepackage{url}
\usepackage{framed}
\usepackage{float}
\usepackage{rotating}
\usepackage{verbatim}
\usepackage{listings}
\usepackage{lscape}

 \usepackage{pgfplots}
 \usepackage{pgf}
 \usepackage{tikz}
 \usetikzlibrary{arrows,shapes.misc,chains,scopes}
 \pgfplotsset{compat = newest}
 \usepackage{pgfplotstable}
 \usepackage{booktabs}
 \usepackage{colortbl}

\newcommand{\executeiffilenewer}[3]{%
\ifnum\pdfstrcmp{\pdffilemoddate{#1}}%
{\pdffilemoddate{#2}}>0%
{\immediate\write18{#3}}\fi%
}

\graphicspath{{figures/}}

\theoremstyle{plain}
\newtheorem{proposition}{Proposition}
\newtheorem{corollary}{Corollary}

\newtheorem{lemma}{Lemma}

\newcounter{algocount}
\newcounter{examplecount}

\newenvironment{example}{\refstepcounter{examplecount}\begin{trivlist}\item \textbf{Example \theexamplecount.}}{\end{trivlist}}
\newenvironment{algorithm}[1][]{\refstepcounter{algocount}\setlength{\parindent}{0.5cm}\begin{trivlist}\item \textbf{Algorithm \thealgocount.}#1\\[-0.2cm]\rule{\columnwidth}{1pt}}{\\[-0.2cm]\rule{\columnwidth}{1pt}\end{trivlist}}

\newcommand{\veczero}{\boldsymbol{0}}

\newcommand{\vecc}{\boldsymbol{c}}

\newcommand{\vecp}{\boldsymbol{p}}

\newcommand{\vect}{\boldsymbol{t}}

\newcommand{\setl}{\ensuremath{\mathcal{L}}\xspace}

\newcommand{\bpm}{\begin{pmatrix}}
\newcommand{\epm}{\end{pmatrix}}
\newcommand{\bbm}{\begin{bmatrix}}
\newcommand{\ebm}{\end{bmatrix}}

\DeclareMathOperator*{\argmin}{argmin}
\DeclareMathOperator*{\argmax}{argmax}

\DeclareMathOperator{\miop}{\mathbb{I}}
\DeclareMathOperator{\kl}{\mathbb{D}}

\DeclareMathOperator*{\minimize}{minimize}

\DeclareMathOperator*{\st}{subject\;to}

\usepackage[noadjust]{cite}
\usepackage{url}

\usepackage[textwidth=3cm,textsize=footnotesize,disable]{todonotes}

\title{Optimal Quantization for Distribution Synthesis\todo{change title to ``Two algos for finite precision approx. of pmfs.''?}}

\author{Georg B\"ocherer and Bernhard C. Geiger
\thanks{The authors contributed equally to this work.}
\thanks{The work of Georg B\"ocherer was partly supported by the German Ministry of Education and Research in the framework of an Alexander von Humboldt Professorship. This paper was presented in part at the 9th International ITG Conference on Systems, Communications and Coding \cite{Boecherer_SCC13}.}
\thanks{Georg B\"ocherer is with the Institute for Communications Engineering, Technische Universit\"at M\"unchen. Email: \texttt{georg.boecherer@tum.de}}
\thanks{Bernhard C. Geiger was with the Signal Processing and Speech Communication Laboratory, Graz University of Technology. He is now with the Institute for Communications Engineering, Technische Universit\"at M\"unchen. Email: \texttt{geiger@ieee.org}}
}

\newcommand{\cra}[1]{\textcolor{black}{#1}}
\newcommand{\crb}[1]{\textcolor{black}{#1}}
\newcommand{\crr}[1]{\textcolor{black}{#1}}

\newcommand{\qtvd}{t^{\mathrm{vd}}}
\newcommand{\vecqtvd}{\boldsymbol{t}^{\mathrm{vd}}}
\newcommand{\qtid}{t^{\mathrm{id}}}
\newcommand{\vecqtid}{\boldsymbol{t}^{\mathrm{id}}}
\newcommand{\vecqt}{\hat{\boldsymbol{t}}}

\newcommand{\vecqtvdt}{\tilde{\vect}^{\mathrm{vd}}}
\newcommand{\qtvdt}{\tilde{t}^{\mathrm{vd}}}

\newcommand{\oleq}[1]{\overset{\text{(#1)}}{\leq}}

\newcommand{\oeq}[1]{\overset{\text{(#1)}}{=}}
\newcommand{\ogeq}[1]{\overset{\text{(#1)}}{\geq}}

\begin{document}

\maketitle

\begin{abstract}
Finite precision approximations of discrete probability distributions are considered, \crb{applicable for distribution synthesis, e.g., probabilistic shaping.} Two algorithms are presented that find the optimal $M$-type approximation $Q$ of a distribution $P$ in terms of the variational distance $\lVert Q-P\rVert_1$ and the informational divergence $D(Q\Vert P)$. Bounds on the approximation errors are derived and shown to be asymptotically tight. Several examples illustrate that the variational distance optimal approximation can be quite different from the informational divergence optimal approximation.
\end{abstract}


\section{Introduction}

Probabilistic models are often used for information processing. In practice, such models are represented with finite precision, e.g., discrete probabilities are represented by rational numbers with finitely many digits. If each probability can be written as an integer multiple of $1/M$ for some integer $M$, then the resulting distribution is called an $M$-\emph{type distribution}. The integer $M$ characterizes the precision by which the rational distribution approximates the true distribution. Additionally, $M$ influences the space needed to store the rational distribution and the complexity to process it. This work studies approximating target distributions $\vect=(t_1,t_2,\dots)$ by $M$-type distributions. 

\subsection{Quality-of-Synthesis Criteria}
One way to measure how good $\vecqt$ approximates $\vect$ is the variational distance
\begin{align}
\lVert\vect-\hat{\vect}\rVert\cra{_1}=\sum_i |t_i-\hat{t}_i|\label{eq:def:vd}
\end{align}
which is symmetric in its arguments $\vect,\hat{\vect}$. Another criterion is the informational divergence
\begin{align}
\kl(\hat{\vect}\Vert\vect)=\sum_{i\colon \hat{t}_i>0} \hat{t}_i\log\frac{\hat{t}_i}{t_i}\label{qos}
\end{align}
where the expectation is taken w.r.t. the approximating distribution $\vecqt$.
The informational divergence with exchanged order of arguments is
\begin{align}
\kl(\vect\Vert\hat{\vect})=\sum_{i\colon t_i>0} t_i\log\frac{t_i}{\hat{t}_i}.\label{qos:dc}
\end{align}
Note that the expectation in \eqref{qos:dc} is taken with respect to the target distribution $\vect$. The informational divergence is asymmetric, i.e., \eqref{qos} and \eqref{qos:dc} are different in general.

In this work we are interested in the scenario where the approximating distribution $\vecqt$ \emph{synthesizes} the distribution $\vect$, i.e, we take expectation with respect to the approximating distribution $\vecqt$. We will therefore consider \eqref{eq:def:vd} and \eqref{qos} as quality-of-synthesis criteria. Several rationales for this choice are as follows:

\subsubsection{Empirical Probability} In distribution synthesis, the approximation $\vecqt$ is the ``true'' distribution and describes a random experiment where the random variable $I$ takes on the integer values $1,2,3,\dots$ according to $\vecqt$, i.e., $\Pr(I=i)=\hat{t}_i$. Denote by $i_1,i_2,\dotsc,i_m$ the outcomes of performing the random experiment $m$ times. \cra{By the law of large numbers,}
\begin{align}
\frac{\sum_{j=1}^m\log\frac{\hat{t}_{i_j}}{t_{i_j}}}{m}\approx\kl(\vecqt\Vert\vect).
\end{align}
There is no such interpretation for the measures\cra{~\eqref{eq:def:vd} and~\eqref{qos:dc}}.
\subsubsection{Infinite Support}
Many important probability distributions have infinite support, e.g., Poisson, Boltzmann, Borel, and Yule-Simon distributions. $M$-type distributions have finite support, and if the target distribution $\vect$ has infinite support, then the measure \eqref{qos:dc} is infinity. The measures\cra{~\eqref{eq:def:vd} and~\eqref{qos}} do not have this issue.
\subsubsection{Probabilistic Shaping\todo{special case of 1)}} Suppose the target distribution $\vect$ is the capacity-achieving input distribution of some communication channel and suppose further that $\vecqt$ is the actual input distribution generated by a communication system. Denote by $W$ the transition probability matrix of the channel. The mutual information $\miop(\vecqt,W)$ that results from using the approximation $\vecqt$ at the channel input is bounded as\todo{Why not directly maximize mutual info over $M$-type distributions, why only looking for the best approximation?}
\begin{align}
\crb{\mathsf{C}\ge}\miop(\vecqt,W)&\oeq{a}\mathsf{C}-\kl(\vecqt W\Vert \vect W)\notag\\
&\ogeq{b}\mathsf{C}-\kl(\vecqt\Vert \vect)\label{qos:motivation}
\end{align}
where $\vecqt W$ and $\vect W$ are the output distributions that result from the input distributions $\vecqt$ and $\vect$, respectively, and where $\mathsf{C}$ is the capacity of the channel. The equality in (a) follows by \crb{\cite[Sec.~III]{bocherer2011matching}},\cite[Proposition~3.11]{bocherer2012capacity} and (b) by the data processing inequality \cite[Lemma~3.11]{csiszar2011information}. The bound \eqref{qos:motivation} shows that as \eqref{qos} approaches zero, the mutual information $\miop(\vecqt,W)$ approaches capacity.

\subsection{Related Work}

For probabilistic shaping, Gallager suggested in \cite[p.~208]{gallager1968information} to choose $\vecqt$ as an $M$-type approximation of the capacity-achieving distribution $\vect$. Several works propose to use \emph{dyadic distributions} in Gallager's scheme, which are $M$-type distributions where $M$ is an integer power of two and where every probability can be written as $2^k/M$ for some integer $k$. The authors in \cite{raphaeli2004constellation} calculate a dyadic approximation by rounding the entries of $\vect$, which minimizes the variational distance \eqref{eq:def:vd}. The authors in \cite{yankov2014rate} calculate the dyadic approximation of $\vect$ that minimizes \eqref{qos} by Geometric Huffman Coding \cite{bocherer2011matching}. Gallager's scheme also works for $M$-type distributions that are not dyadic. In \cite{schreckenbach2005signal}, the authors calculate an $M$-type distribution by a sub-optimal algorithm that aims at minimizing \eqref{qos:dc}. In \cite{Boecherer_SCC13}, we proposed to use $M$-type distributions that minimize \eqref{qos} in Gallager's scheme.

The authors in \cite{Reznik_Quantization,Reznik_QuantizationSPIE} propose a quantization algorithm that minimizes the variational distance \eqref{eq:def:vd}, the Euclidean distance, and the $L_\infty$ norm. The authors also use a Taylor series approximation to analyze their algorithm in terms of the informational divergence \eqref{qos:dc} for \cra{$M$ significantly larger than the support size of the distribution}.

Resolution coding uses an $M$-type input distribution to approximate a target output distribution \cite{han1993approximation}. For the identity channel, \cite[Sec.~III.B]{han1993approximation} constructs an $M$-type approximation that is asymptotically optimal for the variational distance \eqref{eq:def:vd}. In \cite[Sec.~VI.A]{bocherer2013fixed}, informational divergence \eqref{qos} optimal $M$-type approximations are constructed. The authors in \cite{steinberg1996simulation} derive fundamental limits of resolution coding for the identity channel with respect to various approximation measures including \eqref{eq:def:vd} and a normalized version of \eqref{qos}. For noisy channels, resolution coding with respect to variational distance \eqref{eq:def:vd} is considered in \cite{han1993approximation}, informational divergence \eqref{qos} is considered in \cite{hou2013informational} and a normalized version of \eqref{qos} is considered in \cite{wyner1975common,han1993approximation}. Most of the work presented in \cite{wyner1975common,han1993approximation,steinberg1996simulation,hou2013informational} focuses on fundamental limits, i.e., the existence of asymptotically optimal $M$-type approximations is shown but no practical algorithms to construct them are provided.

\subsection{Contributions and Outline}
We propose two simple algorithms that find the $M$-type approximations $\vecqtid$ and $\vecqtvd$ minimizing the informational divergence \eqref{qos} and the variational distance \eqref{eq:def:vd}, respectively. We provide bounds on the approximation errors for target distributions with finite and countably infinite supports. The bounds are asymptotically tight, i.e., any target distribution can be approximated arbitrarily well by an $M$-type approximation with sufficiently large $M$. In Sec.~\ref{sec:discussion}, we show that variational distance \eqref{eq:def:vd} and informational divergence \eqref{qos} lead to fundamentally different $M$-type approximations. In particular, we provide an example where the variational distance optimal approximation $\vecqtvd$ results in an informational divergence equal to one for arbitrarily large $M$. Furthermore, we show that the informational divergence minimizing approximation $\vecqtid$ can have significantly smaller support size than the variational distance minimizing approximation $\vecqtvd$.

\section{Preliminaries}
\label{sec:preliminaries}
Let $\vect$ be a target probability distribution with a finite or countably infinite support. We denote by $n$ the support size of $\vect$. If the support is infinite, then $n=\infty$. Without loss of generality, we assume that $\vect$ is ordered so that $t_1\geq t_2\geq\cdots$. We define the complement of the cumulative distribution function as
\begin{align}
T_k:=\sum_{i>k}t_i\label{eq:Tdef}.
\end{align}
Let $M$ be a positive integer. A distribution $\vecp$ is $M$-type, if each entry can be written as $p_i=c_i/M$ for some non-negative integer $c_i \le M$. We want to determine the $M$-type distribution $\vecp$ that best approximates the target distribution $\vect$. Two quality measures for approximation are considered, namely, the informational divergence and the variational distance as defined in \eqref{qos} and \eqref{eq:def:vd}, respectively. Pinsker's inequality \cite[Lem.~11.6.1]{cover2006elements} bounds the informational divergence from below in terms of the variational distance:
 \begin{equation}
  \lVert \vecp-\vect\rVert_1 \le \sqrt{2 \kl(\vecp\Vert\vect) }.\label{eq:pinsker}
 \end{equation}
There have been several works on bounding the informational divergence from above in terms of the variational distance; see~\cite{Sason_ReversePinsker} for a recent improvement and an overview over available bounds. The most useful for our purposes is adapted from~\cite{Verdu_ITABound}:
\begin{lemma}[{\cite[Thm.~7]{Verdu_ITABound}}]\label{lem:IDboundVD}
For two probability distributions $\vecp$ and $\vect$,
\begin{align}
\kl(\vecp\Vert \vect)\leq \frac{1}{2}\frac{r\log r}{r-1}\lVert \vecp-\vect\rVert_1  \label{eq:pinskerreverse}
\end{align}
where $\displaystyle r:=\sup_{i\colon p_i>0}\frac{p_i}{t_i}\ge 1$.
\end{lemma}
\cra{In Lemma~\ref{lem:IDboundVD} and throughout the remainder of this work, $\log$ denotes the natural logarithm.}

Note that the upper bound \eqref{eq:pinskerreverse} depends on the distributions not only via the variational distance $\lVert \vecp-\vect\rVert_1$, but also via $r$. We therefore call \eqref{eq:pinskerreverse} \emph{distribution dependent}. Any reverse Pinsker's inequality must be distribution dependent, see \cite[Sec.~I.A]{berend2014minimum}. Note further that Lemma~\ref{lem:IDboundVD} was refined in~\cite[Thm.~1]{Sason_ReversePinsker}.


\section{Variational Distance Optimal Quantization}
\label{sec:vd}

\begin{figure}
\begin{algorithm}[ Variational distance optimal approximation.]\label{alg:vd} 
 \\
\cra{Initialize $\vecqtvd=\mathbf{0}$\\
Compute $\qtvd_i\leftarrow\frac{\lfloor Mt_i\rfloor}{M}$, $i=1,\dotsc,\min\{n,M\}$.\\
Compute $e_i\leftarrow t_i-\qtvd_i$, $i=1,\dotsc,\min\{n,M\}$.\\
Compute $L\leftarrow M - M\cdot\sum_{i=1}^{\min\{n,M\}} \qtvd_i$.}\\
\textbf{repeat} $L$ times\\
\indent Choose $\displaystyle j=\min \{\argmax_i e_i\}$. //\emph{choose the smallest index first.} \\
\indent Update $\qtvd_j\leftarrow \qtvd_j+\frac{1}{M}$.\\
\indent Update $e_j\leftarrow t_j-\qtvd_j$.\\
\textbf{end repeat}\\
Return $\vecqtvd$.
\end{algorithm}
\end{figure}

\subsection{Algorithm 1}

An $M$-type approximation of a target distribution $\vect$ can be calculated as follows. First, round off the entries of $\vect$ and then distribute the remaining mass among the entries with the largest error. We call this method \emph{Algorithm 1}, see the top of Page~\pageref{alg:vd}. 

Formally, we first calculate the pre-approximation
\begin{equation}
 \tilde{t}^\text{vd}_i=\frac{\lfloor Mt_i\rfloor}{M},\quad i=1,\dots,n.\label{eq:pre-approximation}
\end{equation}
\cra{Note that in Algorithm~\ref{alg:vd} we can restrict this computation to the first $\min\{n,M\}$ indices of $\vect$ since, by assumption, $\vect$ is ordered, and since not more than $M$ masses can be distributed. Thus, if $n>M$, we can be sure that $\qtvd_i=0$ for $i>M$.}

In general, the entries of $\tilde{\vect}^\text{vd}$ do not sum to one. The pre-approximation gives rise to the non-negative errors
\begin{align}
e_i:=t_i-\tilde{t}^\text{vd}_i\ge 0,\quad i=1,\dotsc,n\label{eq:def:error}
\end{align}
which sum to the \emph{rest mass}
\begin{align}
\sum_{i=1}^n e_i =\frac{L}{M}\label{eq:restmass}
\end{align}
for some integer $L$. Note that the rest mass is bounded as $0\leq L\leq M$, and it is equal to zero if and only if the target distribution $\vect$ is itself $M$-type. 
\begin{example}
For the \cra{2-type} target distribution $\vect=(\frac{1}{2},\frac{1}{2})$ and $M=2$, we have $\tilde{\vect}^\text{vd}=\vect$ and rest mass $0$, i.e., $L=0$. For the \cra{3-type} target distribution $\vect=(\frac{1}{3},\frac{1}{3},\frac{1}{3})$ and $M=2$, we have $\tilde{\vect}^\text{vd}=(0,0,0)$ and rest mass $1$, i.e., $L=M$.
\end{example}
Let $\mathcal{L}$ be \cra{a} set of the indices with the $|\mathcal{L}|=L$ largest error terms, i.e., we have
\begin{align}
i\in\setl,\,j\notin\setl\Rightarrow e_i\geq e_j.\label{eq:setl}
\end{align}
We distribute the remaining $L$ unit masses to the indices in $\setl$, i.e., we choose
\begin{align}
\qtvd_i = \begin{cases}
\tilde{t}_i^\text{vd}+\frac{1}{M},&i\in\setl\\
\tilde{t}_i^\text{vd},& \text{otherwise}.\label{eq:adhoc}
\end{cases}
\end{align}
Note that the set $\mathcal{L}$ is not unique, and consequently the approximation $\vecqtvd$ is not unique either. We illustrate this by an example.
\begin{example}
Let $\vect=(\frac{3}{4},\frac{1}{4})$ and suppose $M=2$. Then
\begin{align}
\tilde{t}^\text{vd}_1=\frac{1}{2},\quad \tilde{t}^\text{vd}_2=0
\end{align}
and
\begin{align}
e_1=e_2=\frac{1}{4}.
\end{align}
Thus, either $\mathcal{L}=\{1\}$ or $\mathcal{L}=\{2\}$. The corresponding approximations $\vecqtvd=(1,0)$ and $\vecqtvd=(1/2,1/2)$ both lead to the same approximation error, namely $\lVert \vecqtvd-\vect\rVert_1=\frac{1}{2}$.
\end{example}
Algorithm 1 resolves this ambiguity by taking entries with lower indices first. From now on, $\vecqtvd$ denotes the unique $M$-type approximation of $\vect$ that is calculated by Algorithm~\ref{alg:vd}. 
\subsection{Elementwise Properties}
From \eqref{eq:def:error} and \eqref{eq:adhoc}, we see that for each index $i$, we have
\begin{equation}
  |t_i-\qtvd_i| < \frac{1}{M}\label{eq:quantBounds}
\end{equation}
and $\vecqtvd$ is a \emph{uniform approximation} of $\vect$. Also by \eqref{eq:def:error} and \eqref{eq:adhoc}, it follows that the approximation $\vecqtvd$ assigns no mass to entries of $\vect$ that are equal to zero, i.e., we have
\begin{equation}
 t_i=0\quad\Rightarrow\quad \qtvd_i = 0.\label{eq:probZero}
\end{equation}
Furthermore, if $\vecqtvd$ assigns zero mass to some entry $t_i$, then it also assigns zero mass to all entries smaller than $t_i$:
\begin{lemma}\label{lem:firstM_vd}
 $t_j<t_i$ and $\qtvd_i=0$ $\Rightarrow$ $\qtvd_j=0$.
\end{lemma}
\begin{proof}
Assume $t_j<t_i$. \crr{In the pre-approximation step, Algorithm~\ref{alg:vd} ensures that $\qtvd_i \ge \frac{\lfloor Mt_i \rfloor}{M}$, hence $\qtvd_i=0$ implies $1/M>t_i>t_j$. Thus, the errors after pre-approximation satisfy $e_i=t_i$, $e_j=t_j$, and $e_j<e_i$. Algorithm~\ref{alg:vd} can only assign a remaining unit mass to $t_j$ and not to $t_i$ if $e_j\geq e_i$. Whence, $\qtvd_j=0$.}
\end{proof}
To prove the optimality of Algorithm~\ref{alg:vd}, we make use of the following lemma.
\begin{lemma}\label{lem:vdopt_quantBounds}
Let $\vect$ be a target distribution with finite or countably infinite support and let $M$ be a positive integer. Every $M$-type approximation $\vecp$ of $\vect$ that is optimal w.r.t. the variational distance satisfies \eqref{eq:quantBounds}.
\end{lemma}
\begin{IEEEproof}
 See Section~\ref{proof:vd:optprop}.
\end{IEEEproof}

\subsection{Optimality of Algorithm 1 and Performance Bounds}
\begin{proposition}\label{prop:vdopt} Let $\vect$ be an ordered target distribution with finite or countably infinite support and let $M$ be a positive integer. Among all $M$-type distributions $\vecp$, $\vecp=\vecqtvd$ minimizes $\lVert\vecp-\vect\rVert_1$.
\end{proposition}
\begin{IEEEproof}
According to Lemma~\ref{lem:vdopt_quantBounds}, any optimal approximation satisfies~\eqref{eq:quantBounds}. Hence, any optimal approximation $\vecp^*$ can be written as
\begin{align}
p^*_i = \begin{cases}
\tilde{t}_i^\text{vd}+\frac{1}{M},&i\in\setl'\\
\tilde{t}_i^\text{vd},& \text{otherwise}
\end{cases}
\end{align}
where $\tilde{\vect}^\text{vd}$ is the pre-approximation \eqref{eq:pre-approximation} and where $\setl'$ is some set of indices with $|\setl'|=L$, where $L$ is given by \eqref{eq:restmass}. We have
\begin{equation}
\lVert\vecp^*-\vect\rVert_1 = \sum_{i\in\setl'} \left(\frac{1}{M}-e_i\right) + \sum_{i\notin\setl'} e_i\label{eq:vdresidual}
\end{equation}
where the error terms $e_i$ are defined in \eqref{eq:def:error}. The residual \eqref{eq:vdresidual} is minimized if $\setl'$ consists of the indices of the $L$ largest error terms $e_i$. According to \eqref{eq:setl}, the approximation calculated by Algorithm~1 has this property.
\end{IEEEproof}
We next bound the variational distance in terms of $M$. If the target distribution $\vect$ has finite support of cardinality $n$, then
\begin{align}
\lVert\vecqtvd-\vect\rVert_1&=\sum_{i=1}^n|\qtvd_i-t_i|\notag\\
&\oleq{a}\sum_{i=1}^n\frac{1}{M}\notag\\
&=\frac{n}{M}\label{eq:vd:simplebound}
\end{align}
where (a) follows by \eqref{eq:quantBounds}. For $n=\infty$, the bound \eqref{eq:vd:simplebound} is infinity for any finite $M$. Thus, we need a different approach to derive a useful bound for the case of infinite support. The next lemma lets us tighten bound \eqref{eq:vd:simplebound} if $M\geq n$ and it will also lead to a useful bound for $n=\infty$. The underlying observation is that we can apply Algorithm~\ref{alg:vd} also to a \emph{sub-probability} distribution, i.e., a target vector whose entries are positive and sum to a value less than or equal to one.
\begin{lemma}\label{lem:subprobability}
 Let $\vect$ be an ordered sub-probability distribution with $k\le M$ entries and total mass $1-T_k$, and let $M$ be a positive integer. Then we have
\begin{align}
 \lVert\vecqtvd-\vect\rVert_1 &
\begin{cases}
 \leq\frac{k}{2M}+\frac{MT_k^2}{2k}, & \text{always}\\
 =T_k, &\text{ if }  T_k\ge\frac{k}{M}                               
\end{cases}\label{eq:subbounds:cases}\\
&\le \frac{k}{2M}+T_k.\label{eq:subbounds:loose}
\end{align}
\end{lemma}
Note that for $T_k=k/M$ both cases in \eqref{eq:subbounds:cases} coincide.
\begin{IEEEproof}
The proof is given in Sec.~\ref{sec:proof:lem}.
\end{IEEEproof}
\cra{A distribution can be split into two sub-probability distributions, one containing the first $k$ indices, and one containing the tail of the distribution. More specifically, we can split $\vect$ into two vectors $\vect_{1:k}$ and $\vect_{\mathrm{tail}}$ with the same length but disjoint support sets: The entries of $\vect_{1:k}:=(t_1,\dots,t_k,0,0,0,\dots)$ are zero for indices larger than $k$, while for $\vect_{\mathrm{tail}}:=(0,0,\dots,t_{k+1},\dots,t_n)$ the first $k$ entries are zero. Let $\vecqtvd_{1:k}$ denote the approximation that results from applying Algorithm~\ref{alg:vd} to $\vect_{1:k}$}. We have
\begin{equation}
  \lVert\vecqtvd_{1:k}-\vect\rVert_1 = \lVert\vecqtvd_{1:k}-\vect_{1:k}\rVert_1 + T_k\label{eq:divide_and_conquer}
\end{equation}
where $\lVert\vecqtvd_{1:k}-\vect_{1:k}\rVert_1$ can be bounded by Lemma~\ref{lem:subprobability}. This divide-and-conquer approach is useful when the number of entries of the target distribution exceeds the type $M$ of the approximating distribution. Approach \eqref{eq:divide_and_conquer} is also used in the proof of the following proposition, which states various bounds on the approximation error of $\vecqtvd$.
\begin{proposition}\label{prop:vdbounds}%
Let $\vect$ be an ordered target distribution and let $M$ be a positive integer. 
\begin{enumerate}
\item If $\vect$ has finite support of cardinality $n\leq M$, then
\begin{align}
\lVert\vecqtvd-\vect\rVert_1\leq \frac{n}{2M}.\label{eq:vd:mgen}
\end{align}
\item If $\vect$ has finite or countably infinite support of cardinality $n> M$, then
\begin{align}
\lVert\vecqtvd-\vect\rVert_1\leq \frac{k}{2M}\left( 1+\frac{MT_k}{k}\right)^2\le \frac{2k}{M}\label{eq:vd:mlen}
\end{align}
 where $k$ is the support size of $\vecqtvd$.
\item For $n=\infty$, the support size $k$ of $\vecqtvd$ satisfies $k\stackrel{M\to\infty}{\longrightarrow} \infty$ and $k/M\stackrel{M\to\infty}{\longrightarrow} 0$.
\end{enumerate}
\end{proposition}
\begin{IEEEproof}
The proof is given in Sec.~\ref{sec:proof:vd}.
\end{IEEEproof}
We next give examples that illustrate the tightness of the bounds.
\begin{example}
For $n<\infty$, the bound \eqref{eq:vd:mgen} is tight for a uniform target distribution and $M=3n/2$. For $M<n$, the bound \eqref{eq:vd:mlen} is tight for, e.g., $M=5$ and $t_1=t_2=t_3=4/15$ and $t_i<1/15$ for all $i>3$ ($n$ arbitrary).
\end{example}
\subsection{Asymptotic Optimality}
For target vectors with finitely many entries, the bound \eqref{eq:vd:mgen} guarantees that the approximation error of $\vecqtvd$ can be made arbitrarily small by choosing $M$ large enough. The same is true for infinitely many entries. This follows by bound \eqref{eq:vd:mlen} together with Statement 3) of Proposition~\ref{prop:vdbounds}. Furthermore, by \eqref{eq:quantBounds} the $M$-type approximation converges uniformly to the target distribution. We summarize these observations as a corollary to Proposition~\ref{prop:vdbounds}.
\begin{corollary}\label{cor:vd:asymptotic}
Let $\vect$ be an ordered target distribution with finite or countably infinite support. For $M\to\infty$, the approximation $\vecqtvd$ converges uniformly to the target distribution $\vect$.
\end{corollary}
For $M\ge n$ the variational distance decreases with $\mathcal{O}(1/M)$. For $M<n$ no such convergence guarantee can be given. This is illustrated in the next example.
\begin{example}\label{ex:vd:yulesimon}
Consider the Yule-Simon distribution~\cite{simon1955class} with $t_i=\rho B(i,\rho+1)$, where $\rho>0$ and where $B(\cdot,\cdot)$ is the beta-function. Lemma~\ref{lem:firstM_vd} ensures that Algorithm~\ref{alg:vd} assigns unit masses to at most the first $M$ indices. For $M>1$, we have
\begin{align}
\lVert\vecqtvd-\vect\rVert_1 &= \sum_{i=1}^\infty |\qtvd_i-t_i| \ge T_M\notag\\
&=M B(M,\rho+1)\label{eq:yulesimonccdf}\\
&\geq \frac{K(\rho)}{(M+\rho+1)^\rho}\label{eq:yulesimonbound}
\end{align}
where $K(\rho)$ is a positive constant that does not depend on $M$, see Sec.~\ref{proof:yulesimonbound} for the derivation. Thus, the convergence of Algorithm~\ref{alg:vd} is at best $\mathcal{O}(1/M^\rho)$.
\end{example}

\section{Informational Divergence Optimal Quantization}
\label{sec:id}
\begin{figure}
\begin{algorithm}[ Informational divergence optimal quantization.]\label{alg:id}\ 
\\
Initialize $c_i\leftarrow 0$, $i=1,\dotsc,n$.\\
\textbf{for} $m=1,2,\dotsc,M$\\
\indent Choose $\displaystyle j=\min\{\argmin_i \Delta_i(c_i+1)\}$. //\emph{choose the smallest index first.}\\
\indent Update $c_j\leftarrow c_j+1$.\\
\textbf{end for}\\
Return $\vecc$.
\end{algorithm}
\end{figure}
We now consider $M$-type quantization with respect to the informational divergence, i.e., we want to solve the problem
\begin{align}
\begin{split}
\minimize_{\vecp}\quad&\kl(\vecp\Vert\vect)\\
\st\quad&\vecp\text{ is $M$-type}.
\end{split}\label{prob:KLDproblem}
\end{align}

\subsection{Equivalent Problem}
Recall that each entry $p_i$ of an $M$-type distribution can be written as $p_i=c_i/M$ for some non-negative integer $c_i$. We have
\begin{align}
\kl(\vecp\Vert\vect)&=\sum_{i\colon c_i>0} \frac{c_i}{M}\log\frac{\frac{c_i}{M}}{t_i}\nonumber\\
&=\frac{1}{M}\Bigl(\sum_{i\colon c_i>0} c_i\log\frac{c_i}{t_i}\Bigr)-\log M
\end{align}
so that Problem~\eqref{prob:KLDproblem} is equivalent to
\begin{align}
\begin{split}
\minimize_{c_1,\dotsc,c_n}\quad&\sum_{i\colon c_i>0} c_i\log\frac{c_i}{t_i}\\
\st\quad&c_i\in\{0,1,2,\dotsc,M\},\quad i=1,\dotsc,n\\
&\sum_{i=1}^n c_i = M.
\end{split}\label{prob:mod}
\end{align}
If $\vecc^*$ is a solution of Problem~\eqref{prob:mod}, then $\vecp^*=\vecc^*/M$ is a solution of Problem \eqref{prob:KLDproblem}.

\subsection{Algorithm~\ref{alg:id}}
To solve problem~\eqref{prob:mod}, we write the objective function as a telescoping sum
\begin{align}
\sum_{i\colon c_i>0} c_i\log\frac{c_i}{t_i}&=\sum_{i=1}^{n}\sum_{k=1}^{c_i}\Bigl[k\log\frac{k}{t_i}-(k-1)\log\frac{k-1}{t_i}\Bigr]\nonumber\\
&=\sum_{i=1}^{n}\sum_{k=1}^{c_i}\Delta_i(k)
\end{align}
where the increment function is
\begin{align}
\Delta_i(k)=k\log k-(k-1)\log(k-1)+\log\frac{1}{t_i}.\label{eq:id:increment}
\end{align}
\crb{Evaluating $\Delta_i(x)$ as a function of a real number $x$ and taking the derivative,
\begin{align}
\frac{\partial}{\partial x}\Delta_i(x)=\log\frac{x}{x-1},\label{eq:id:derivative}
\end{align}
we conclude that $\Delta_i(k)$ is strictly monotonically increasing in $k$. Moreover, rewriting~\eqref{eq:id:increment} as
\begin{align}
\Delta_i(k)&=k\log \frac{k}{k-1}+\log(k-1)+\log\frac{1}{t_i}\notag\\
&\geq\log(k-1).\label{eq:id:incrementbound}
\end{align}
(which holds trivially for $k=1$) shows that the increment function grows without bound with $k$. The following lemma summarizes the properties of the increment function.
\begin{lemma}\label{lem:ifproperties}
For all $m>0$, the increment function $\Delta_i(k)$ grows without bound with $k$ and satisfies
\begin{align}
&\ell>m\Rightarrow\Delta_i(\ell)>\Delta_i(m)\label{eq:id:lkproperty}\\
&t_i>t_j\Rightarrow\Delta_i(m)<\Delta_j(m).\label{eq:id:ijproperty}
\end{align}
\end{lemma}
}


An allocation $\vecc$ can be obtained by initially assigning the zero vector $\veczero$ to a pre-allocation $\tilde{\vecc}$ and successively incrementing the entry of $\tilde{\vecc}$ by one for which the corresponding increment cost $\Delta(\tilde{c}_i+1)$ is smallest. After $M$ iterations, the constraint $\sum_i \tilde{c}_i=M$ is fulfilled and $\vecc=\tilde{\vecc}$ is a valid allocation. If more than one entry of $\tilde{\vecc}$ has the smallest increment cost in some step, then either of them can be chosen, so the allocation obtained by this strategy is not unique. We illustrate this by the following example.
\begin{example}\label{ex:id2sols}
Suppose $\vect=(\frac{4}{5},\frac{1}{5})$ and $M=2$. We have $\Delta_1(1)=\log\frac{5}{4}$ and $\Delta_2(1)=\log 5$, so after the first step, $\tilde{\vecc}=(1,0)$. In the second step, we have
\begin{align}
\Delta_1(2)=2\log(2)+\log\frac{5}{4}=\log 5,\quad \Delta_2(1)=\log 5,
\end{align}
so the final allocation is either $\vecc_1=(2,0)$ or $\vecc_2=(1,1)$. The corresponding approximations are $\vecp_1=(1,0)$ and $\vecp_2=(\frac{1}{2},\frac{1}{2})$. Both approximations lead to the same informational divergence, namely
\begin{align}
\kl(\vecp_1\Vert \vect)=\kl(\vecp_2\Vert\vect)=\log\frac{5}{4}.
\end{align}
\end{example}
Algorithm~\ref{alg:id} resolves this ambiguity by incrementing entries with lower index first. From now on, we denote by $\vecqtid$ the unique $M$-type approximation of $\vect$ that is calculated by Algorithm~\ref{alg:id}.

\subsection{Elementwise Properties}

The informational divergence is a weighted sum of $\log\frac{\qtid_i}{t_i}$. We therefore expect that for a good approximation $\vecqtid$, the ratio $\qtid_i/t_i$ is close to one. The next lemma states this property.
\begin{lemma}\label{lem:ratiobound}
Let $\vect$ be a target distribution with finite or countably infinite support and let $M$ be a positive integer. Every $M$-type approximation $\vecp$ of $\vect$ that is optimal w.r.t. the informational divergence satisfies
 \begin{align}
\frac{p_i}{t_i} < \frac{e}{t_1},\quad \forall i\leq k\label{eq:id:ratiobound}
\end{align}
where $k$ is the support size of $\vecp$. In particular
\begin{align}
\frac{1}{Mt_k}\le \frac{e}{t_1}.\label{eq:id:ratioboundk}
\end{align}
\end{lemma}

\begin{IEEEproof}
See Section~\ref{sec:proof:ratiobound}.
\end{IEEEproof}

Lemma~\ref{lem:ratiobound} directly implies
\begin{align}
t_i=0\Rightarrow\qtid_i=0.\label{eq:id:zeromass}
\end{align}
Furthermore, if $\vecqtid$ assigns zero mass to some entry $t_i$, then it also assigns zero mass to all entries smaller than $t_i$:
\begin{lemma}\label{lem:firstM_id}
 $t_j<t_i$ and $\qtid_i=0$ $\Rightarrow$ $\qtid_j=0$.
\end{lemma}
\begin{IEEEproof}
The statement follows by~\eqref{eq:id:ijproperty} for $m=1$.
\end{IEEEproof}
\subsection{Optimality and Performance Bounds}
\begin{proposition}\label{prop:idopt}
Let $\vect$ be an ordered target distribution with finite or countably infinite support and let $M$ be a positive integer. Among all $M$-type distributions $\vecp$, $\vecp=\vecqtid$ minimizes $\kl(\vecp\Vert\vect)$.
\end{proposition}
\begin{IEEEproof}
 See Section~\ref{sec:proof:idopt}.
\end{IEEEproof}
The increment in the $m$-th iteration of Algorithm~\ref{alg:id} does not depend on $M$. This means that the algorithm not only calculates the optimal $M$-type quantization, but actually \emph{all} optimal $m$-type quantizations for $m=1,2,\dotsc,M$. We state this property as a corollary of Proposition~\ref{prop:idopt}.
\begin{corollary}\label{cor:id}
Let $\vecc$ be the pre-allocation calculated by Algorithm~\ref{alg:id} in the $m$-th iteration and define
\begin{align*}
\vecqtid_m:=\left(\frac{c_1}{m},\dotsc,\frac{c_n}{m}\right).
\end{align*}
Among all $m$-type distributions $\vecp$, $\vecp=\vecqtid_m$ minimizes $\kl(\vecp\Vert\vect)$.
\end{corollary}
We next bound the informational divergence in terms of $M$. We start with the case when the support size of the target distribution is finite $(n<\infty)$. We have
\begin{align}
\kl(\vecqtid\Vert\vect)&\oleq{a}\kl(\vecqtvd\Vert\vect)\notag\\
&\oleq{b}\sum_{i\colon \qtvd_i>0} \qtvd_i\left(\frac{\qtvd_i}{t_i}-1\right)\notag\\
&\oleq{c}\sum_{i\colon \qtvd_i>0} \qtvd_i\left(\frac{t_i+\frac{1}{M}}{t_i}-1\right)\notag\\
&\leq\frac{1}{t_nM}\label{eq:id:simplebound}
\end{align}
where (a) follows by the optimality of $\vecqtid$, (b) by $\log(x)\leq x-1$, and (c) by \eqref{eq:quantBounds}. For $n=\infty$, we have $t_i\overset{i\to\infty}{\to}0$, so bound \eqref{eq:id:simplebound} becomes useless.  The next proposition tightens \eqref{eq:id:simplebound} for $n<\infty$ and $M\geq n$ and it provides a bound for $M<n$, which is important when the support of $\vect$ is infinite.
\begin{proposition}\label{prop:idbounds}
Let $\vect$ be an ordered target distribution and let $M$ be a positive integer.
\begin{enumerate}
\item If $\vect$ has finite support of cardinality $n\leq M$, then 
\begin{align}
\kl(\vecqtid\Vert\vect) < \log\left(1+\frac{n}{2 t_n M^2}\right).\label{eq:id:mgen}
\end{align}
\item If $\vect$ has finite or countably infinite support of cardinality $n>M$, then
\begin{align}
&\kl(\vecqtid\Vert\vect)< \frac{1}{2}\frac{r \log r}{r-1}\left(\frac{k}{2M}+2T_k\right)\label{eq:id:mlen}
\end{align}
with $r=\frac{1}{1-T_k}+\frac{e}{t_1}$.
\item For $n=\infty$, the support size $k$ of $\vecqtid$ satisfies $k\stackrel{M\to\infty}{\longrightarrow} \infty$ and $k/M\stackrel{M\to\infty}{\longrightarrow} 0$.
\end{enumerate}
\end{proposition}
\begin{IEEEproof}
See the Section~\ref{sec:proof:id}.
\end{IEEEproof}	
We briefly discuss the intuition behind the bounds in Proposition~\ref{prop:idbounds}.
The bound \eqref{eq:id:mgen} follows by evaluating the informational divergence of the variational distance optimal approximation $\vecqtvd$. To derive bound \eqref{eq:id:mlen}, we apply Lemma~\ref{lem:IDboundVD}. First, we determine the support size $k$ of $\vecqtid$. Then, we use Algorithm~\ref{alg:vd} to approximate the sub-probability distribution $\vect_{1:k}$. This lets us bound both the ratio $r$ and the variational distance in Lemma~\ref{lem:IDboundVD}. Note that \eqref{eq:id:mgen} and \eqref{eq:id:mlen} are not tight for finite $M$.

\subsection{Asymptotic Optimality}

For target distributions with finite support, bound \eqref{eq:id:mgen} guarantees that the informational divergence can be made arbitrarily small by choosing $M$ large enough. This result is also valid for target distributions with infinite support by using Statement 3) of Proposition~\ref{prop:idbounds} in \eqref{eq:id:mlen}. We summarize these observations as a corollary to Proposition~\ref{prop:idbounds}.
\begin{corollary}\label{cor:id:asymptotic}
Let $\vect$ be an ordered target distribution with finite or countably infinite support. For $M\to\infty$, the informational divergence of $\vecqtid$ and $\vect$ approaches zero.
\end{corollary} 
For $M\ge n$, the informational divergence approaches zero as $\mathcal{O}(1/M^2)$ by bound \eqref{eq:id:mgen}. For $M<n$, no such speed of convergence guarantee can be stated. We illustrate this by the following example.
\begin{example}
By Lemma~\ref{lem:firstM_id}, $\vecqtid$ assigns mass only to \cra{at most} the first (largest) $M$ indices. As in Example~\ref{ex:vd:yulesimon}, we consider the Yule-Simon distribution. By Pinsker's inequality \eqref{eq:pinsker} and Example~\ref{ex:vd:yulesimon}, the convergence of Algorithm~\ref{alg:id} is at best $\mathcal{O}(1/M^{2\rho})$. 
\end{example}

\section{Comparison of Informational Divergence and Variational Distance}
\label{sec:discussion}

\subsection{Elementwise Properties}

The variational distance optimal approximation $\vecqtvd$ guarantees a bounded per-entry approximation error $|t_i-\qtvd_i|$ by \eqref{eq:quantBounds}. Correspondingly, the informational divergence optimal approximation $\vecqtid$ guarantees a bounded per-entry ratio $\qtid_i/t_i$ by \eqref{eq:id:ratiobound}. The approximations $\vecqtvd$ and $\vecqtid$ can violate the per-entry bounds of the other. We illustrate this by the following two examples.
\begin{example}
Let $t_1=1/M$ and $t_2=\cdots=t_{n}=\frac{M-1}{(n-1)M}$, for $n>M$. Hence $\vecqtvd=(\frac{1}{M},\dots,\frac{1}{M})$, and
\begin{equation}
 \frac{\qtvd_2}{t_2} = \frac{(n-1)M}{M(M-1)} =\frac{n-1}{M-1}
\end{equation}
can be arbitrarily large. The approximation $\vecqtid$ guarantees that, by \eqref{eq:id:ratiobound}, we have
\begin{align}
\frac{\qtid_2}{t_2}\leq \frac{e}{t_1}=eM
\end{align}
independent of $n$.
\end{example}

\begin{example}
Let $\vect=(0.97,\,0.01,\,0.01,\,0.01)$ and $M=256$. It follows that $L=2$ and we obtain $\vecqtvd=(248,\,3,\,3,\,2)/256$ from Algorithm~\ref{alg:vd}. Algorithm~\ref{alg:id}, however, yields $\vecqtid=(247,\,3,\,3,\,3)/256$, where
\begin{equation}
 t_1-\qtid_1 = \frac{1.32}{M}
\end{equation}
violates~\eqref{eq:quantBounds}.

Let $\vect=(0.4,\,\varepsilon,\,\varepsilon,\dots,\,\varepsilon)^T$ and $M=2$. It follows that $L=2$ and we obtain $\vecqtvd=(1/2,\,1/2,\dots,\,0,\,0)^T$ from Algorithm~\ref{alg:vd}. However, if $n$ is sufficiently large such that $\varepsilon<0.1$, it can be shown that Algorithm~\ref{alg:id} yields $\vecqtid=(1,\,0,\dots,\,0,\,0)^T$, where
\begin{equation}
 \qtid_1-t_1 = \frac{1.2}{M}
\end{equation}
violates~\eqref{eq:quantBounds}.
\end{example}

\subsection{Support}

Suppose the target distribution $\vect$ has infinite support. By Statement 3) in Proposition~\ref{prop:vdbounds} and Statement 3) in Proposition~\ref{prop:idbounds}, the supports of the approximations $\vecqtvd$ and $\vecqtid$ both increase without bound and sublinearly with $M$. However, the following example shows that the support of $\vecqtvd$ can grow much faster than the support of $\vecqtid$. The reason is that assigning probability masses to indices with small target probabilities has a much higher cost in terms of informational divergence than in terms of variational distance. We illustrate this phenomenon by the following example.
\begin{example}
Consider the Yule-Simon distribution (see Example~\ref{ex:vd:yulesimon}) with $\rho=0.2$ and let $M$ take values from 1 to 10000 in steps of 10. The resulting support sizes of  $\vecqtvd$ and $\vecqtid$ are displayed in Fig.~\ref{fig:support}. The support size of $\vecqtvd$ is around twice the support size of $\vecqtid$. The considered Yule-Simon distribution has a heavy tail with $T_{10000} \approx 0.15$. In other words, the first 10000 entries of $\vect$ contain only 85\% of the total probability mass. 
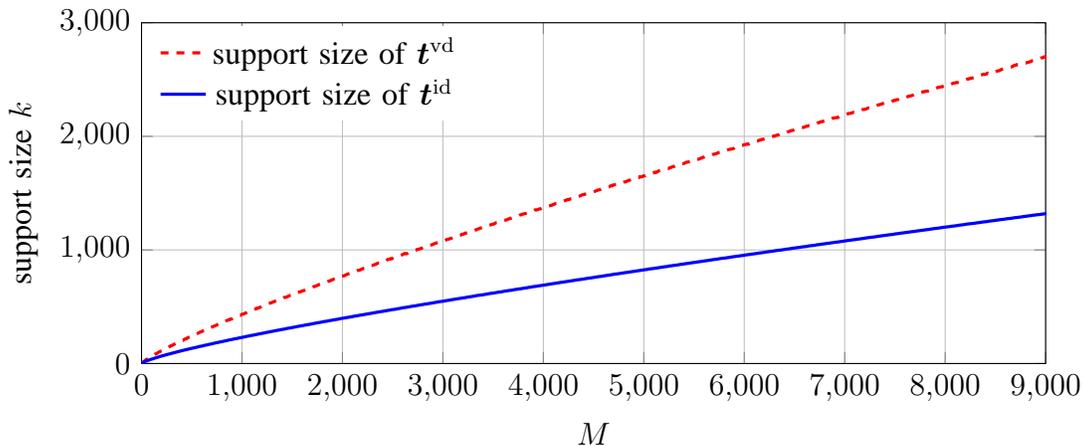
\begin{figure}[t]
 \centering

\begin{tikzpicture}
\begin{axis}[
width=0.75\textwidth,
height=0.25\textheight,
xmin = 0,
ymin = 0,
xmax = 9000,
ymax = 3000,
xlabel={$M$},
ylabel={support size $k$},
grid=both,
xlabel near ticks,
ylabel near ticks,
legend columns=1,
legend style={at={(0.01,0.99)},anchor=north west,draw=none},
legend entries = {{support size of $\vecqtvd$},{support size of $\vecqtid$}},
]
\addplot[red,no markers,very thick,solid,dashed]
table{kvd.dat};
\addplot[blue,no markers,very thick,solid]
table{kid.dat};
\end{axis}

\end{tikzpicture}
 \caption{Support sizes of $\vecqtvd$ and $\vecqtid$ for the Yule-Simon distribution with $\rho=0.2$.}
 \label{fig:support}
\end{figure}
\end{example}

The next example shows that the support of $\vecqtvd$ is not always larger than the support of $\vecqtid$.


\begin{example}
In Example~\ref{ex:id2sols} we showed that for $\vect=(4/5,1/5)$ and $M=2$ both $\vecqt_1=(1,0)$ and $\vecqt_2=(1/2,1/2)$ are optimal in terms of the informational divergence. As it can be easily shown, $\vecqt_1$ is the unique approximation that is optimal in terms of the variational distance. We now modify the target distribution to $\vect=(4/5-\epsilon,1/5+\epsilon)$ with $0<\epsilon<1/20$. The vector $\vecqt_1$ remains the unique variational distance optimal approximation and $\vecqt_2$ is now the unique informational divergence optimal approximation. The support of $\vecqt_2$ is strictly larger than the support of $\vecqt_1$.
\end{example}

\subsection{Asymptotic Optimality}

Corollaries \ref{cor:vd:asymptotic} and \ref{cor:id:asymptotic} state that $\vecqtvd$ and $\vecqtid$ are asymptotically optimal w.r.t. variational distance and informational divergence, respectively. By Pinsker's inequality \eqref{eq:pinsker}, $\vecqtid$ is also asymptotically optimal w.r.t. the variational distance. In contrast, the variational distance optimal approximation $\vecqtvd$ is in general not asymptotically optimal w.r.t. the informational divergence. This is illustrated by the following example.
\begin{example}
Consider the distribution $\vect$ that is constructed from the geometric distribution $\tilde{t}_i=2^{-i}$ as follows: First, $t_1=\tilde{t}_1$. Then, the next probability mass $\tilde{t}_2$ is split into so many pieces that for $M=2$ the informational divergence equals $\cra{\log 2}$. For $M=2$, Algorithm~\ref{alg:vd} yields $\vecqtvd_2 = (\frac{1}{2},\frac{1}{2})$, where the first \cra{entry} is approximated perfectly. The informational divergence of $\vecqtvd_2$ and $\vect$ evaluates to
\begin{equation}
 \kl(\vecqtvd_2\Vert\vect) = \frac{1}{2}\log \frac{1}{2t_2} \stackrel{!}{=} \cra{\log 2}
\end{equation}
from which $t_2=1/8$ follows. Thus, $t_2=t_3=1/8$, which sums to $1/4$. Repeating the procedure for $M=8$, the first three indices are approximated without error, and the two remaining masses are placed on the following indices, such that $\vecqtvd_8 = (\frac{1}{2},\frac{1}{8},\frac{1}{8},\frac{1}{8},\frac{1}{8})$. To ensure that the informational divergence remains equal to $\cra{\log 2}$, one again must split the next probability mass $\tilde{t}_3$ into sufficiently many pieces. It follows that $t_4=\cdots=t_{19}=1/128$, which sum to $1/8$. Repeating this procedure yields $\vect$ satisfying
\begin{multline}
 t_j = 2^{1-2^i},\quad\\ \text{ if }
 \sum_{k=0}^{i-1} 2^{2^k-k-1} \le j \le \sum_{k=0}^{i} 2^{2^k-k-1} -1,\\ i\in\mathbb{N}.
\end{multline}
For this, the subsequence $\{M_i\}_{i\in\mathbb{N}}=\{2^{2^i-1}\}$ yields an informational divergence equal to $\cra{\log 2}$, while the variational distance is bounded by $2/2^i$, i.e., twice the remaining mass of the geometric distribution. Hence, by Corollary~\ref{cor:vd:asymptotic}, $\lVert\vecqtvd-\vect\rVert_1\overset{M\to\infty}\to 0$, while $\limsup_{M\to\infty}\kl(\vecqtvd\Vert\vect)=\cra{\log 2}$.
\end{example}

 \section{Proofs}
\label{sec:proofs}
%
%
%
%
%
%
%
%
%
\subsection{Proof of Lemma~\ref{lem:vdopt_quantBounds}}~\label{proof:vd:optprop}
\crb{We prove that every optimal $\vecp$ satisfies~\eqref{eq:quantBounds} by contradiction: Suppose that $p_i\le t_i-\frac{1}{M}$ for some $i$. Since both $t_i$ and $p_i$ must sum to one, there must be a $j$ for which $p_j>t_j$. Define $\vecp^\circ$ by $p_i^\circ=p_i+\frac{1}{M}$, $p_j^\circ = p_j-\frac{1}{M}$, and $p_\ell^\circ=p_\ell$ for all $\ell\neq i,j$.
We calculate
\begin{align}
\lVert\vecp-\vect\rVert_1-\lVert\vecp^\circ-\vect\rVert_1
&=t_i-p_i-t_i+p_i^\circ+|p_j-t_j|-|p_j^\circ-t_j| \notag\\
&=\frac{1}{M}+|p_j-t_j|-|p_j-\frac{1}{M}-t_j|\\
&=\frac{1}{M}+|p_j-t_j|-\left||p_j-t_j|-\frac{1}{M}\right|\\
&>0.\label{eq:vdproof:lastline}
\end{align}
where~\eqref{eq:vdproof:lastline} follows because $p_j>t_j$. We conclude that an optimal algorithm cannot lead to $p_i\le t_i-\frac{1}{M}$. That $p_i\ge t_i+\frac{1}{M}$ is sub-optimal follows along the same lines.\qed
}


\subsection{Proof of Lemma~\ref{lem:subprobability}}
\label{sec:proof:lem}
We claim that the two bounds in \eqref{eq:subbounds:cases} relate as 
\begin{align}
T_k\leq\frac{k}{2M}+\frac{MT_k^2}{2k}.\label{vd:proofs:sub0} 
\end{align}
This can be seen from
\begin{align}
\left(\frac{k}{2M}+\frac{MT_k^2}{2k}\right)-T_k&=\frac{M}{2k}\left(\frac{k^2}{M^2}-2\frac{k}{M}T_k+T_k^2\right)\notag\\
&=\frac{M}{2k}\left(\frac{k}{M}-T_k\right)^2\geq 0.
\end{align}
The general bound in~\eqref{eq:subbounds:loose} follows by loosening the right-hand side (left-hand side) of~\eqref{vd:proofs:sub0} if $T_k\le k/M$ (if $T_k\ge k/M$).

We next consider the two cases $T_k \geq k/M$ and $T_k \leq k/M$ separately.

\textbf{Case $T_k \geq k/M$:} We show that $\lVert\vecqtvd-\vect\rVert_1= T_k$ and the general bound follows by \eqref{vd:proofs:sub0}. We have
\begin{align}
\frac{k}{M}\leq T_k &= 1-\sum_{i=1}^k t_i=\sum_{i=1}^k (\qtvd_i - t_i).\label{vd:proofs:sub1}
\end{align}
\cra{In Algorithm~\ref{alg:vd}, the rest mass $L/M$ after the initialization step cannot be smaller than $T_k$. Thus
\begin{align}
\frac{L}{M}\geq T_k\geq \frac{k}{M}
\end{align}
which implies $L\geq k$. Thus, in the finalization step of Algorithm~\ref{alg:vd}, each of the entries $j=1,\dotsc,k$ will get assigned at least one more mass $1/M$, so
\begin{align}
\text{for\,each }j=1,\dotsc,k\colon (\qtvd_i-t_i)\geq 0.\label{vd:proofs:sub2}
\end{align}
Altogether, we have
\begin{equation}
 \lVert\vecqtvd-\vect\rVert_1 = \sum_{i=1}^k |\qtvd_i-t_i| \oeq{a} \sum_{i=1}^k (\qtvd_i-t_i) \oeq{b} T_k\label{vd:proofs:sub3}
\end{equation}
where (a) follows by \eqref{vd:proofs:sub2} and where (b) follows by \eqref{vd:proofs:sub1}.
}

\textbf{Case $T_k\leq k/M$:} If $\qtvd_i-t_i\geq0$ for all $i=1,\dotsc,k$, then $\lVert\vecqtvd-\vect\rVert_1 = T_k$ by \eqref{vd:proofs:sub3} and \eqref{vd:proofs:sub0} implies that the general bound claimed by the lemma holds. It remains to show that the general bound also holds when
\begin{align}
\qtvd_j-t_j<0 \text{ for some }j\label{vd:proofs:sub4} 
\end{align}
\cra{
which implies
\begin{align}
\qtvd_i-t_i<\frac{1}{M},\quad i=1,\dotsc,k.
\end{align}
In particular, \eqref{vd:proofs:sub4} implies $L<k$ for the rest mass after the initialization step in Algorithm~\ref{alg:vd}, which implies further that in the finalization step, each entry $i=1,\dotsc,k$ gets assigned at most one additional mass $1/M$.} The error mass after the initialization step is
\begin{align}
 \sum_{i=1}^k e_i &= \sum_{i=1}^kt_i-\sum_{i=1}^k\frac{\lfloor Mt_i\rfloor}{M}\notag\\
&=\frac{L}{M}-T_k.
\end{align}
Now reorder the $k$ errors such that $\tilde{e}_i\ge\tilde{e}_{i+1}$. We bound the mean error from below and above by
\begin{equation}
 \frac{1}{L} \sum_{i=1}^L \tilde{e}_i \ge \frac{L}{Mk}-\frac{T_k}{k} \ge \frac{1}{k-L}\sum_{i=L+1}^k \tilde{e}_i. \label{eq:vdproof:inequality}
\end{equation}
Equality holds if $\tilde{e}_i=\frac{L}{Mk}-\frac{T_k}{k}$ for all $i=1,\dots,k$. After the update step in Algorithm~\ref{alg:vd}, the $L$ largest errors $\tilde{e}_i$ are replaced by the final errors $1/M-\tilde{e}_i$. The other errors remain unchanged. We bound 
\begin{align}
 \sum_{i=1}^k |\qtvd_i-t_i|
&= \sum_{i=1}^L \left(\frac{1}{M}-\tilde{e}_i\right) + \sum_{i=L+1}^k \tilde{e}_i\notag\\
&\oleq{a} \frac{L}{M} + \left(\frac{L}{Mk}-\frac{T_k}{k}\right)(k-2L)
\end{align}
where (a) follows by \eqref{eq:vdproof:inequality}. The maximum is achieved for $L=(k+MT_k)/2$, which yields
\begin{equation}
 \lVert\vecqtvd-\vect\rVert_1 \le \frac{k}{2M}+\frac{MT_k^2}{2k}.\label{vd:proofs:subfinal}
\end{equation}
\qed

\subsection{Proof of Proposition~\ref{prop:vdbounds}}
\label{sec:proof:vd}

\subsubsection{} 
The proof follows from Lemma~\ref{lem:subprobability} for $k=n$ and $T_k=T_n\equiv 0$.

\subsubsection{}
Let $k$ be the support size of $\vecqtvd$, and let $\vect_{1:k}$ be the sub-probability distribution obtained by taking the first $k$ indices of $\vect$. Then, we have
\begin{equation}
 \lVert\vecqtvd-\vect\rVert_1 = \lVert\vecqtvd-\vect_{1:k}\rVert_1 + T_k.
\end{equation}
If $k$ is the support size, then by Lemma~\ref{lem:firstM_vd} the first $k$ indices get $M$ masses. Since the algorithm satisfies~\eqref{eq:quantBounds}, we have
\begin{align}
 T_k = 1-\sum_{i=1}^k t_i
=\sum_{i=1}^k (\qtvd_i-t_i)
\le \frac{k}{M}.
\end{align}
Thus we can bound $\lVert\vecqtvd-\vect_{1:k}\rVert_1$ by Lemma~\ref{lem:subprobability} and get
\begin{align}
  \lVert\vecqtvd-\vect\rVert_1 
&\le \frac{k}{2M}+\frac{MT_k^2}{2k} + T_k\notag\\
&= \frac{k}{2M} \left( 1+\frac{2MT_k}{k} +\frac{M^2T_k^2}{k^2}\right)\notag\\
&= \frac{k}{2M}\left( 1+\frac{MT_k}{k}\right)^2.
\end{align}

\subsubsection{}
The support size $k$ of $\vecqtvd$ grows without bound with $M$ because for every $l$ there exists an $M$ such that $t_l>1/M$, hence this index gets probability mass already in the initialization step of Algorithm~\ref{alg:vd}.

We show that the support size $k\equiv k(M)$ grows sublinearly with $M$ by contradiction. Suppose there exists a $0<c\le 1$ such that
\begin{equation}
 \limsup_{M\to\infty} \frac{k(M)}{M} = c.
\end{equation}
Thus, for each $\epsilon>0$, there exists a sequence $\{M_i\}_{i\in\mathbb{N}}$, $M_1<M_2<M_3<\dotsb$, such that
\begin{equation}\label{eq:subseqbounds}
 (c-\epsilon)M_i < k(M_i) < (c+\epsilon)M_i,\quad i\in\mathbb{N}.
\end{equation}
Now choose $i<j\in\mathbb{N}$. Applying the algorithm for $M_i$ and $M_j$ increases the support size from $k(M_i)$ to $k(M_j)$. In total, the algorithm has $M_j$ masses to distribute, some of which are distributed to the first $k(M_i)$ indices. In particular, in the first step the algorithm assigns
\begin{equation}
 \sum_{l=1}^{k(M_i)} \lfloor M_j t_l\rfloor
\end{equation}
masses to the first $k(M_i)$ indices. The difference in support sizes is thus bounded from above by
\begin{align}
 k(M_j)&-k(M_i)
\le M_j - \sum_{l=1}^{k(M_i)} \lfloor M_j t_l\rfloor\notag\\
&< M_j - \sum_{l=1}^{\lfloor(c-\epsilon)M_i\rfloor} \lfloor M_j t_l\rfloor\notag\\
&= M_j \left( 1-\sum_{l=1}^{\lfloor(c-\epsilon)M_i\rfloor} \frac{\lfloor M_j t_l\rfloor}{M_j}\right)\notag\\
&= M_j \left(T_{\lfloor(c-\epsilon)M_i\rfloor}+\sum_{l=1}^{\lfloor(c-\epsilon)M_i\rfloor} \left(t_l-\frac{\lfloor M_j t_l\rfloor}{M_j}\right)\right)\notag\\
&< M_j \left( T_{\lfloor(c-\epsilon)M_i\rfloor}+\frac{(c-\epsilon)M_i}{M_j}\right).
\end{align}
Now choose $i$ large enough such that $T_{\lfloor(c-\epsilon)M_i\rfloor}< \epsilon$ and choose $j$ large enough such that $M_i/M_j<1/4$. We have
\begin{equation}
 \frac{k(M_j)-k(M_i)}{M_j} < \epsilon+\frac{c-\epsilon}{4} .\label{eq:proofs:subupper}
\end{equation}
A lower bound on the support size difference is obtained from~\eqref{eq:subseqbounds}:
\begin{equation}
  \frac{k(M_j)-k(M_i)}{M_j}
> (c-\epsilon)-(c+\epsilon)\frac{M_i}{M_j}
> \frac{3c}{4} - \frac{5\epsilon}{4}.\label{eq:proofs:sublower}
\end{equation}
Combining \eqref{eq:proofs:subupper} and \eqref{eq:proofs:sublower} yields an upper bound on $c$:
\begin{equation}
 \frac{3c}{4} - \frac{5\epsilon}{4} < \frac{c-\epsilon}{4} +\epsilon.
\end{equation}
After rearranging we have $c<4\epsilon$ for any $\epsilon>0$, and thus
\begin{equation}
 \limsup_{M\to\infty} \frac{k(M)}{M} = 0.
\end{equation}

\qed

\subsection{Proof of \eqref{eq:yulesimonbound}}\label{proof:yulesimonbound}
\cra{We make use of the following lower bound on the beta function~\cite[eq.~(2)]{Grenie_BetaBounds}
\begin{equation}
 B(x,y)\ge \frac{x^{x-1}y^{y-1}}{(x+y)^{x+y-1}}
\end{equation}
which in our case gives
\begin{align}
 M\cdot B(M,\rho+1) &\ge \frac{M^M (\rho+1)^\rho}{(M+\rho+1)^{M+\rho}}\notag\\
 &= \frac{M^M }{(M+\rho+1)^M}\frac{(\rho+1)^\rho}{(M+\rho+1)^{\rho}}\notag\\
 &= \frac{(\rho+1)^\rho}{(1+\frac{\rho+1}{M})^M} \frac{1}{(M+\rho+1)^{\rho}}\notag\\
 &\ge \frac{(\rho+1)^\rho}{e^{\rho+1}} \frac{1}{(M+\rho+1)^{\rho}}\label{eq:rhobound}
\end{align}
where~\eqref{eq:rhobound} follows because  $(1+\frac{\rho+1}{M})^M$ approaches $e^{\rho+1}$ from below. This shows the existence of the constant $K(\rho)$ in \eqref{eq:yulesimonbound}.
}

\subsection{Proof of Lemma~\ref{lem:ratiobound}}
\label{sec:proof:ratiobound}
\cra{The case $M=1$ (hence $k=1$) is trivial; we focus on $M\ge 2$.}
 Suppose that $\vecp$ is an $M$-type distribution (not necessarily optimal) and that $\vecp^\circ$ is such that $p_i^\circ=p_i+\frac{1}{M}\le 1$, $p_j^\circ=p_j-\frac{1}{M}\ge 0$ and $p_\ell=p^\circ_\ell$ for all $\ell\neq i,j$. We now show that $\kl(\vecp\Vert\vect) > \kl(\vecp^\circ\Vert\vect)$ holds if $\vecp$ violates the statement of Lemma~\ref{lem:ratiobound}, i.e., that $\vecp$ is not optimal is not optimal in this case. To this end, notice that

 \begin{align*}
  \kl(\vecp\Vert\vect)-\kl(\vecp^\circ\Vert\vect)
  &= p_i \log \frac{p_i}{t_i} + p_j \log \frac{p_j}{t_j} - \left(p_i+\frac{1}{M}\right) \log \frac{p_i+\frac{1}{M}}{t_i} - \left(p_j-\frac{1}{M}\right) \log \frac{p_j-\frac{1}{M}}{t_j}\\
  &\stackrel{(a)}{=} p_j \log \frac{p_j}{t_j}- \left(p_j-\frac{1}{M}\right) \log \frac{p_j-\frac{1}{M}}{t_j} - \frac{1}{M}\left(\Delta_i(Mp_i+1)-\log M\right)\\
  &\stackrel{(b)}{>} \frac{1}{M}\log \frac{p_j}{t_j} + \underbrace{\left(p_j-\frac{1}{M}\right)\log\frac{p_j}{p_j-\frac{1}{M}}}_{>0}- \frac{1}{M}\left(\Delta_i(M)-\log M\right)\\
  &> \frac{1}{M}\log\frac{p_j}{t_j}+ \underbrace{\frac{M-1}{M}\log\frac{M-1}{M}}_{\ge -\frac{1}{M}} + \frac{1}{M}\log t_i\\
  &\ge \frac{1}{M}\log\frac{p_j}{t_j} - \frac{1}{M}\log \frac{e}{t_i}
 \end{align*}
 where $(a)$ is due to~\eqref{eq:id:increment} and $(b)$ follows by \eqref{eq:id:lkproperty}. Hence, if
 \begin{equation}\label{eq:proof:ratiobound}
  \frac{p_j}{t_j} \ge \frac{e}{t_i}
 \end{equation}
 for any pair of indices $i$ and $j$, then above difference of informational divergences is positive as well. Thus, an optimal $\vecp$ may not fulfill~\eqref{eq:proof:ratiobound} for any such pair of indices. The best bound is obtained for $i=1$, hence Lemma~\ref{lem:ratiobound} follows.
%
The result for index $k$ results from $p_k\ge 1/M$.\qed

\subsection{Proof of Proposition~\ref{prop:idopt}}
\label{sec:proof:idopt}

To prove optimality, we need the following lemma.
\begin{lemma}\label{lem:bounded}
Let $\vecc^*$ be an optimal allocation. Let $\vecc$ be a pre-allocation with $\sum_ic_i<M$ and $c_i\leq c^*_i$ for $i=1,\dotsc,n$. Define
\begin{align}
j=\argmin_i\Delta_i(c_i+1).\label{eq:defj}
\end{align}
Then there exists an optimal allocation $\tilde{\vecc}$ with
\begin{align}
c_j+1&\leq \tilde{c}_j\label{eq:lem2:statement}\\
c_i&\leq\tilde{c}_i,\quad i=1,\dotsc,n.
\end{align}
\end{lemma}
\begin{proof}
Suppose we have
\begin{align}
c_j+1>c^*_j.\label{eq:lem2:supp}
\end{align}
Since $c_j\leq c^*_j$ by assumption, \eqref{eq:lem2:supp} implies
\begin{align}
c_j+1=c^*_j+1.\label{eq:lem2:prop1}
\end{align}
Since $\sum_ic_i<M$ and $\sum_ic^*_i=M$, there must be at least one $\ell\neq j$ with
\begin{align}
c^*_\ell\geq c_\ell+1.\label{eq:lem2:prop2}
\end{align}
By decreasing $c^*_\ell$ by one and increasing $c^*_j$ by one, the change of the objective function is $\Delta_j(c^*_j+1)-\Delta_\ell(c^*_\ell)$. We bound this change as follows:
\begin{align}
\Delta_j(c^*_j+1)-\Delta_\ell(c^*_\ell)&\overset{\text{(a)}}{\leq} \Delta_j(c^*_j+1)-\Delta_\ell(c_\ell+1)\label{eq:lem1:neq1}\\
&\overset{\text{(b)}}{=} \Delta_j(c_j+1)-\Delta_\ell(c_\ell+1)\notag\\
&\overset{\text{(c)}}{\leq} \label{eq:lem1:neq2}0
\end{align}
where (a) follows by \eqref{eq:lem2:prop2} and Lemma~\ref{lem:ifproperties}, (b) follows by \eqref{eq:lem2:prop1}, and (c) follows by the definition of $j$ in \eqref{eq:defj}.

We must consider two cases. First, suppose we have strict inequality in either \eqref{eq:lem1:neq1} or \eqref{eq:lem1:neq2}. Then the objective function is decreased, which contradicts the assumption that $\vecc^*$ is optimal. Thus, the supposition \eqref{eq:lem2:supp} is false and the statements of the lemma hold for $\tilde{\vecc}=\vecc^*$. Second, suppose we have equality both in \eqref{eq:lem1:neq1} and \eqref{eq:lem1:neq2}. In this case, define the allocation
\begin{align}
\tilde{c}_\ell = c^*_\ell-1,\quad\, \tilde{c}_j = c^*_j + 1,\quad\, \tilde{c}_i = c^*_i\text{ for } i\neq j,\ell.
\end{align}
Equality in \eqref{eq:lem1:neq1}--\eqref{eq:lem1:neq2} implies optimality of $\tilde{\vecc}$. By \eqref{eq:lem2:prop1} and \eqref{eq:lem2:prop2}, we can verify that $\tilde{\vecc}$ fulfills the statements of the lemma. This concludes the proof of Lemma~\ref{lem:bounded}.
\end{proof}
We are now ready to prove Proposition~\ref{prop:idopt}. By Lemma~\ref{lem:bounded}, there is an optimal allocation $\tilde{\vecc}$ such that in each iteration of Algorithm~\ref{alg:id} we have 
\begin{align}
c_i\leq \tilde{c}_i,\qquad i=1,\dotsc,n.\label{eq:proof1}
\end{align}
After Algorithm~\ref{alg:id} terminates, we have
\begin{align}
M=\sum_i c_i \leq \sum_i \tilde{c}_i = M.\label{eq:proof2}
\end{align}
Statements \eqref{eq:proof1} and \eqref{eq:proof2} can be true simultaneously only if $c_i= \tilde{c}_i$ for all $i=1,\dotsc,n$. Consequently, the constructed allocation $\vecc$ is optimal.\qed

\subsection{Proof of Proposition~\ref{prop:idbounds}}
\label{sec:proof:id}

\subsubsection{Case $M\ge n$}
By Proposition~\ref{prop:idopt}, $\vecqtid$ is optimal w.r.t. the informational divergence and
\begin{equation}
 \kl(\vecqtid\Vert\vect) \le \kl(\vecqtvd\Vert\vect).
\end{equation}
Moreover,
\begin{align}
 \kl(\vecqtvd\Vert\vect) &= \sum_{i=1}^n\qtvd_i \log\frac{\qtvd_i}{t_i}\notag\\
 &\stackrel{(a)}{\le} \log\left( \sum_{i=1}^n\frac{(\qtvd_i)^2}{t_i}\right)\notag\\
 &= \log\left(1+ \sum_{i=1}^n \frac{(\qtvd_i-t_i)^2}{t_i}\right) \label{eq:renyi}
\end{align}
where $(a)$ is Jensen's inequality (see also the proof of~\cite[Thm.~3]{Sason_ReversePinsker}) and where the sum inside the logarithm is Pearson's $\chi^2$-distance $\chi^2(\vecqtvd\Vert\vect)$. Note that~\eqref{eq:renyi} equals $\kl_2(\vecqtvd\Vert\vect)$, the R\'{e}nyi divergence of second order. The inequality in $(a)$ is then a direct consequence of the fact that R\'{e}nyi divergence is non-decreasing in the order~\cite[Thm.~3]{vanErven_KLD}.

We now bound~\eqref{eq:renyi} by
\begin{align}
\sum_{i=1}^n\frac{(\qtvd_i-t_i)^2}{t_i}
&\cra{\le}\frac{1}{t_n}\sum_{i=1}^n(\qtvd_i-t_i)^2\notag\\
&= \frac{1}{t_n} \sum_{i=1}^n |\qtvd_{i}-t_i|\underbrace{|\qtvd_{i}-t_i|}_{< \frac{1}{M}\text{ by }\eqref{eq:quantBounds}}\notag\\
&< \frac{1}{t_n M} \lVert\vecqtvd-\vect\rVert_1\notag\\
&\oleq{a}\frac{n}{2t_n M^2}
\end{align}
where (a) follows by Statement 1) in Proposition~\ref{prop:vdbounds}.


\subsubsection{Case $M<n$} 
Let $k$ be the support size of $\vecqtid$. Define the auxiliary distribution $\tilde{\vect}:=\vect_{1:k}/(1-T_k)$. Because of the normalization by $1-T_k$, the entries of $\tilde{\vect}$ sum to one and $\tilde{\vect}$ is a distribution. Denote by $\vecqtvdt$ the approximation that results from applying Algorithm~\ref{alg:vd} to $\tilde{\vect}$. We have

\begin{align}
\kl(\vecqtid\Vert\vect)\leq\kl(\vecqtvdt\Vert\vect)\oleq{a}\frac{1}{2}\frac{r\log r}{r-1}\lVert \vecqtvdt-\vect\rVert_1\nonumber\\
\text{with }r=\max_{i\leq k}\frac{\qtvdt_i}{t_i}\label{eq:proof:mlenbound}
\end{align}

where (a) follows by Lemma~\ref{lem:IDboundVD}. It remains to bound the ratio $r$ and the variational distance $\lVert \vecqtvdt-\vect\rVert_1$. 

\emph{Bounding $r$:}
By \eqref{eq:quantBounds}, we have
\begin{align}
\qtvdt_i < \tilde{t}_i+\frac{1}{M}=\frac{t_i}{1-T_k}+\frac{1}{M}.
\end{align}
Thus, for each $i\leq k$, we have
\begin{align}
\frac{\qtvdt_i}{t_i} <\frac{1}{1-T_k}+\frac{1}{t_iM}\leq\frac{1}{1-T_k}+\frac{1}{t_kM}
\end{align}
which implies
\begin{align}
r &< \frac{1}{1-T_k}+\frac{1}{t_kM}\notag\\
&\oleq{a}\frac{1}{1-T_k}+\frac{e}{t_1}\label{eq:proof:rbound}
\end{align}
where (a) follows by \eqref{eq:id:ratioboundk} in Lemma~\ref{lem:ratiobound}.
	
\emph{Bounding $\lVert \vecqtvdt-\vect\rVert_1$:}
We bound
\begin{align}
\lVert \vecqtvdt-\vect\rVert_1&=\lVert \vecqtvdt-\vect_{1:k}\rVert_1+T_k\notag\\
&=\lVert \vecqtvdt-\tilde{\vect}(1-T_k)\rVert_1+T_k\notag\\
&\oleq{a}\lVert \vecqtvdt-\tilde{\vect}\rVert_1+\lVert\tilde{\vect}T_k\rVert_1+T_k\notag\\
&=\lVert \vecqtvdt-\tilde{\vect}\rVert_1+2T_k\notag\\
&\leq\frac{k}{2M}+2T_k\label{eq:proof:vdbound}
\end{align}
where (a) follows by the triangle inequality. Using \eqref{eq:proof:rbound} and \eqref{eq:proof:vdbound} in \eqref{eq:proof:mlenbound} completes the proof.

\subsubsection{}
The support $k$ grows without bound because the increment functions $\Delta_i$ grow without bound by \eqref{eq:id:incrementbound}, i.e., for every positive integer $\ell$ there exists an $M$ large enough such that, for all $i=1,\dots,\ell-1$,
\begin{equation}
 \log \frac{1}{t_\ell} < \Delta_i(c_i+1)
\end{equation}
where the sum over all $c_i$ is less than $M$. In other words, after assigning a specific number of masses to indices 1 to $\ell-1$, assigning a mass to index $\ell$ must have lower cost than assigning additional masses to the first $\ell-1$ indices.

The result $k(M)/M\overset{M\to\infty}{\to}0$ can be seen as follows. Increasing $M$ by one increases the support size $k$ at most by one. This is a consequence of the update rule in Algorithm~\ref{alg:id}. Thus, the sequence $k\equiv k(M)$ contains each integer $1,2,3,\dotsc$ at least once and we can define a sequence $M(k)$, $k=1,2,3,\dotsc$. Note that some integers may not occur in the sequence $M(k)$. By \eqref{eq:id:ratioboundk} in Lemma~\ref{lem:ratiobound}, we can bound the $k$-th probability $t_k$ by
\begin{align}
t_k>\frac{t_1}{eM(k)}
\end{align}
and we have
\begin{align}
1=\sum_{k=1}^\infty t_k > \sum_{k=1}^\infty\frac{t_1}{eM(k)}.
\end{align}
If $M(k)$ grows only linearly with $k$, then the sum on the right-hand side diverges, which contradicts that the probabilities need to sum to one. Thus, $M(k)$ grows super-linearly with $k$ and equivalently, $k(M)$ grows sub-linearly with $M$.\qed

\section{Acknowledgment}
The authors thank I. Sason for fruitful discussions (in particular for suggesting Lemma~\ref{lem:IDboundVD}) and R. A. Amjad for pointing out Corollary~\ref{cor:id}. 
The authors are grateful to G. Kramer for helpful comments on drafts.
\bibliographystyle{IEEEtran}
\bibliography{IEEEabrv,confs-jrnls,references}

\listoftodos

\end{document}